\documentclass{elsarticle}

\usepackage{verbatim}
\usepackage[usenames]{color}
\usepackage{amssymb,amsfonts,amsmath,enumerate,latexsym,yhmath,amsthm}
\usepackage{url}
\usepackage{graphicx}

\usepackage{tikz}
\usetikzlibrary{snakes,shapes}
\tikzstyle{hyb}=[rectangle,fill=green!50,draw,minimum size=3mm]
\tikzstyle{tre}=[circle,fill=green!50,draw,minimum size=3.5mm]
\newcommand{\etq}[1]{%
\draw (#1) node {\footnotesize $#1$};
}

 \theoremstyle{plain}
 \newtheorem{theorem}{Theorem}
 \newtheorem{lemma}[theorem]{Lemma}
 
 \newtheorem{proposition}[theorem]{Proposition}

 \theoremstyle{definition}
 \newtheorem{example}[theorem]{Example}

\renewcommand{\leq}{\leqslant}
\renewcommand{\geq}{\geqslant}

\newcommand{\TT}{\mathcal{T}}

\newcommand{\RR}{\mathbb{R}}
\newcommand{\oF}{\overline\Phi}
\newcommand{\oFd}{\overline\Phi^{(2)}}

\begin{document}
\begin{frontmatter}

\title{The expected value of the squared euclidean cophenetic metric under the Yule and the uniform models}

\author{Gabriel Cardona}
\ead{gabriel.cardona@uib.es}
\author{Arnau Mir}
\ead{arnau.mir@uib.es}
\author{Francesc Rossell\'o\corref{cor1}}
\ead{cesc.rossello@uib.es}
\cortext[cor1]{Corresponding author}
\address{Department of Mathematics and  Computer Science,
  University of the Balearic Islands,
  E-07122 Palma de
  Mallorca, Spain}

\begin{abstract}
The cophenetic metrics $d_{\varphi,p}$, for $p\in \{0\}\cup[1,\infty[$, are a recent addition to the kit of available distances for the comparison of phylogenetic trees. Based on a fifty years old idea of Sokal and Rohlf, these metrics compare phylogenetic trees on a same set of taxa by encoding them by means of their  vectors of cophenetic values of pairs of taxa and depths of single taxa, and then computing the $L^p$ norm of the difference of the corresponding vectors.
In this paper we compute the expected value  of the square of $d_{\varphi,2}$  on the space of fully resolved rooted phylogenetic trees with $n$ leaves, under  the Yule and the uniform probability distributions.
 \end{abstract}

\end{frontmatter}

\section{Introduction}

The definition and study of metrics for the comparison of rooted phylogenetic trees  on the same set of taxa is a classical problem in phylogenetics  \cite[Ch.~30]{fel:04}, and many metrics have been introduced so far with this purpose. A recent addition to the set of metrics available in this context are the \emph{cophenetic metrics} $d_{\varphi,p}$ introduced in \cite{copheneticd1}. Based on a fifty years old idea of Sokal and Rohlf, these metrics compare phylogenetic trees on a same set of taxa by  first encoding the trees by means of their vectors of cophenetic values of pairs of taxa and depths of single taxa, and then computing the $L^p$ norm of the difference of the corresponding vectors.

Once the disimilarity between two phylogenetic trees has been computed through a given metric, it is convenient in many situations to assess its signifiance. One possibility is to compare the value obtained with its expected, or mean, value: is it much larger, much smaller, similar? \cite{steelpenny:sb93} This makes it necessary to study the distribution of the metric, or,  at least, to have a formula for the expected  value of the metric for any number $n$ of leaves. The distribution of several metrics has been studied so far: see, for instance, \cite{BS09,Yulenodal,Hendy,MirR10,steelpenny:sb93}.

The expected value of a distance depends on the probability distribution on the space of phylogenetic trees under consideration.
The most popular distribution on the space $\TT_n$ of binary phylogenetic trees with $n$ leaves is the uniform distribution, under which all trees in $\TT_n$ are equiprobable. But phylogeneticists consider also other probability distributions on $\TT_n$, defined through stochastic models of evolution \cite[Ch.~33]{fel:04}. The most popular is the so-called Yule model \cite{Harding71,Yule}, defined by an evolutionary process where, at each step, each currently extant species can give rise,  with the same probability, to two new species. Under this model, different phylogenetic trees with the same number of leaves may have different probabilities, which depend on their shape. 

In this paper we provide explicit formulas for the expected values  under the uniform and the Yule models of the square of the \emph{euclidean cophenetic metric} $d_{\varphi,2}$. 
The proofs of these formulas are based on long and tedious algebraic computations and thus, to ease the task of the reader interested only in the formulas and the path leading to them, but not in the details, we have moved these computations to an Appendix at the end of the paper.

Besides the aforemenentioned application of this value in the assessment of tree comparisons, the knowledge of formulas for the  expected value of $d_{\varphi,2}^2$ under different models may allow the use of $d_{\varphi,2}$ to test stochastic models  of tree growth, a popular line of research in the last years which so far has been mostly based on shape indices; see, for instance, \cite{BF06,Mooers97}. As a proof of concept, in \S 4 we report on a basic, preliminary such test performed on the binary phylogenetic  trees contained in the TreeBASE database \cite{treebase}.

\section{Preliminaries}

In this paper, by a  \emph{phylogenetic tree} on a set $S$ of taxa we mean a fully resolved, or binary, rooted tree  with its leaves bijectively labeled in $S$.  We understand such a rooted  tree as a directed graph, with its arcs pointing away from the root. To simplify the language, we shall always identify a leaf of a phylogenetic tree with its label.  We shall also use the term \emph{phylogenetic tree with $n$ leaves} to refer to a phylogenetic tree on the set $\{1,\ldots,n\}$.  We shall denote by $\TT(S)$ the space of all phylogenetic trees on $S$ and by  $\TT_n$ the space of all phylogenetic trees with $n$ leaves.

Let $T$ be a phylogenetic tree. If there exists a directed path from $u$ to $v$ in  $T$, we shall say that $v$ is a  \emph{descendant} of $u$ and also that $u$ is an  \emph{ancestor} of $v$. The \emph{lowest common ancestor} $\mathrm{LCA}_T(u,v)$ of a pair of nodes $u,v$ in   $T$ is the unique common ancestor of them that is a descendant of every other common ancestor of them.  The \emph{depth}  $\delta_T(v)$ of a node $v$ in $T$ is the distance (in number of arcs) from the root of $T$ to $v$. The  \emph{cophenetic value} $\varphi_T(i,j)$ of a  pair of  leaves $i,j$ in $T$   is the depth of their LCA. To simplify the notations, we shall often write $\varphi_T(i,i)$ to denote the depth $\delta_T(i)$ of a leaf $i$. 

Given two phylogenetic trees $T,T'$ on disjoint sets of taxa $S,S'$, respectively, we shall denote by $T\,\widehat{\ }\, T'$ the phylogenetic tree on $S\cup S'$ obtained by connecting the roots of $T$ and $T'$ to a (new) common root.
Every phylogenetic  tree $T\in \TT_n$ is obtained as $T_k\widehat{\ }\, {}T'_{n-k}$, for some $1\leq k\leq n-1$, some subset  $S_k\subseteq \{1,\ldots,n\}$ with $k$ elements, some tree $T_k$ on $S_k$ and some tree $T'_{n-k}$ on $S_k^c=\{1,\ldots,n\}\setminus S_k$. Actually, every phylogenetic  tree in $\TT_n$ is obtained in this way {twice}.

The \emph{Yule}, or  \emph{Equal-Rate Markov}, model of evolution \cite{Harding71,Yule} is a stochastic model of phylogenetic trees' growth. It starts with a node, and at every step a leaf is chosen randomly and uniformly and it is splitted into two leaves. Finally, the labels are assigned randomly and uniformly to the leaves once the desired number of leaves is reached. This corresponds to a model  of evolution where, at each step, each currently extant species can give rise,  with the same probability, to two new species.  Under this stochastic model, if $T\in \TT_n$ is a phylogenetic tree with set of internal nodes $V_{int}(T)$, and if for every  $v\in V_{int}(T)$ we denote by $\ell_T(v)$ the number of its descendant leaves, then the probability of  $T$ is   \cite{Brown,SM01}
 $$
 P_Y(T)=\frac{2^{n-1}}{n!}\prod_{v\in V_{int}(T)}\frac{1}{\ell_T(v)-1}.
$$

The  \emph{uniform}, or \emph{Proportional to Distinguishable Arrangements},  model \cite{Rosen78} is another stochastic model of phylogenetic trees' growth. Unlike the Yule model, its main feature is that all phylogenetic trees $T\in \TT_n$ have the same probability:
$$
P_U(T)=\frac{1}{(2n-3)!!}\mbox{, where } (2n-3)!!=(2n-3)(2n-5)\cdots 3\cdot 1.
$$
From the point of view of tree growth, this model is described as the process that starts with a node labeled 1 and then, at  the $k$-th step, a new pendant arc, ending in the leaf labeled $k+1$, is added either to a new root (whose other child will be, then, the original root) or to some edge, with all possible locations of this new pendant arc being equiprobable  \cite{CS,cherries}. Although this is not an explicit model of evolution, only of tree growth, several interpretations of it in terms of evolutionary processes have been given in the literature: see \cite[p. 686]{BF06} and the references therein.

\section{Main results}

Let $T\in \TT_n$ be a  phylogenetic tree with $n$ leaves.  The \emph{cophenetic vector} of $T$ is
$$
\varphi(T)=\big(\varphi_T(i,j)\big)_{1\leq i\leq j\leq n}\in \RR^{n(n+1)/2},
$$
with its elements lexicographically ordered in $(i,j)$. 
It turns out \cite{copheneticd1} that the mapping
$\varphi: \TT_n \to \RR^{n(n+1)/2}$
sending each $T\in \TT_n$ to its  cophenetic vector $\varphi(T)$, is  injective up to isomorphism. As it is well known, this allows to induce metrics on  $\TT_n$ from metrics defined on powers of $\RR$. In particular, 
in this paper we consider the cophenetic metric $d_{\varphi,2}$ on $\TT_n$ induced by the euclidean distance:
$$
d_{\varphi,2}(T_1,T_2)= \sqrt{\sum_{1\leq i\leq j\leq
n}(\varphi_{T_1}(i,j)-\varphi_{T_2}(i,j))^2}.
$$
To distinguish it from other cophenetic metrics obtained through other $L^p$ normes, we shall call it the \emph{euclidean cophenetic metric}.

\begin{example}
Consider the phylogenetic trees $T,T'\in \TT_4$ depicted in Fig. \ref{fig:1}.
Their total cophenetic vectors are
$$
\begin{array}{l}
\varphi(T)\hspace*{0.5ex}=(2,1,0,0,2,0,0,2,1,2)\\
\varphi(T')=(1,0,0,0,2,1,1,3,2,3)
\end{array}
$$
and therefore $d_{\varphi,2}(T,T')^2=7$. As we shall see below, the expected values of the square of $d_{\varphi,2}$ on $\TT_4$ under the uniform and the Yule models are, respectively, $10.56$ and $9.41$, and hence these two trees are quite more similar than average with respect to the euclidean cophenetic metric under both models.
\end{example}

\begin{center}
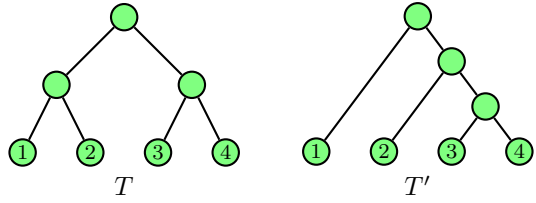
\begin{figure}[htb]
\begin{center}
\begin{tikzpicture}[thick,>=stealth,scale=0.3]
\draw(0,0) node [tre] (1) {};  \etq 1
\draw(3,0) node [tre] (2) {};  \etq 2
\draw(6,0) node [tre] (3) {};  \etq 3
\draw(9,0) node [tre] (4) {};  \etq 4  
\draw(1.5,3) node [tre] (a) {};   
\draw(7.5,3) node [tre] (b) {};   
\draw(4.5,6) node [tre] (r) {};   
\draw (r)--(a);
\draw (r)--(b);
\draw (a)--(1);
\draw (a)--(2);
\draw (b)--(3);
\draw (b)--(4);
\draw(4.5,-1.5) node {$T$};
\end{tikzpicture}
\qquad
\begin{tikzpicture}[thick,>=stealth,scale=0.3]
\draw(0,0) node [tre] (1) {};  \etq 1
\draw(3,0) node [tre] (2) {};  \etq 2
\draw(6,0) node [tre] (3) {};  \etq 3
\draw(9,0) node [tre] (4) {};  \etq 4  
\draw(7.5,2) node [tre] (b) {};   
\draw(6,4) node [tre] (a) {};   
\draw(4.5,6) node [tre] (r) {};   
\draw (r)--(a);
\draw (r)--(1);
\draw (a)--(2);
\draw (a)--(b);
\draw (b)--(3);
\draw (b)--(4);
\draw(4.5,-1.5) node {$T'$};
\end{tikzpicture}
\end{center}
\caption{\label{fig:1} 
Two phylogenetic trees with 4 leaves.}
\end{figure}
\end{center}

Let $D_n^2$ the random variable that chooses a pair of trees $T,T'\in \TT_n$ and computes $d_{\varphi,2}(T,T')^2$.
Its expected values under the Yule and the uniform models are given by the following two theorems. Recall that the $n$-th \emph{harmonic number} $H_n$ is defined as $H_n=\sum_{i=1}^n 1/i$.

\begin{theorem}\label{th:Yule}
For every $n\geq 2$, the expected value of $D_n^2$ under the Yule model is
$$
E_Y(D_n^2)= 
\frac{2n}{n-1}\big(3 n^2-10 n-1+8 (n+1) H_n-4 (n+1) H_n^2\big).
$$
\end{theorem}

\begin{theorem}\label{th:unif}
For every $n\geq 2$, the expected value of $D_n^2$ under the uniform model is
$$
E_U(D_n^2) =
\frac{1}{3}(4n^3+18n^2-10n)-\frac{n(n+3)}{2}\cdot\frac{(2n-2)!!}{(2n-3)!!}-\frac{n(n+7)}{4}\left(\frac{(2n-2)!!}{(2n-3)!!}\right)^2
$$
\end{theorem}

Since $H_n\sim \ln(n)$ and $(2n-2)!!/(2n-3)!!\sim \sqrt{\pi n}$, these formulas imply that
$$
E_Y(D_n^2)\sim 6n^2,\quad E_U(D_n^2)\sim \Big(\frac{4}{3}-\frac{\pi}{4}\Big)n^3.
$$

We shall prove the  formulas in Theorems \ref{th:Yule} and \ref{th:unif} by reducing the computation of the expected value of $D_n^2$ to that of the following random variables:
\begin{itemize}
\item $S_n$, the random variable that chooses a tree $T\in \TT_n$ and computes  its Sackin index $S$ \cite{Sackin:72}, defined by
$$
S(T)=\sum_{i=1}^n \delta_T(i)
$$

\item $\Phi_n$, the random variable that chooses a tree $T\in \TT_n$ and computes its total cophenetic index $\Phi$  \cite{MRR}, defined by 
$$
\Phi(T)=\hspace*{-2ex}\sum_{1\leq i<j\leq n} \varphi_T(i,j)
$$

\item $\oF_n^{(2)}$, the random variable that chooses a tree $T\in \TT_n$ and computes
$$
\oF^{(2)}(T)=\sum\limits_{1\leq i\leq j\leq n}\varphi_T(i,j)^2
$$
\end{itemize}
For the models under consideration, the expected values of these variables are related to that of $D_n^2$ by the next proposition.
In it and henceforth, we shall denote by $E(X)$ the expected value of a random variable $X$ on $\TT_n$ under a generic probability distribution
 $p:\TT_n\to [0,1]$ on $\TT_n$ invariant under relabelings. The probability distributions $p_Y$ and $p_U$ defined by the Yule and the uniform models, respectively, are invariant under relabelings, and therefore the expected values under these specific models, which will be denoted by $E_Y$ and $E_U$, respectively, are  special cases of $E$.
 
\begin{proposition}\label{prop:E}
$E(D_n^2)=2E(\oFd_n)-2\cdot \dfrac{E(S_n)^2}{n}- 4\cdot\dfrac{E(\Phi_n)^2}{n(n-1)}.$
\end{proposition}

\begin{proof}
To simplify the notations, let
\begin{itemize}
\item $\varphi_n$ be the random variable that chooses a tree $T\in \TT_n$ and computes  $\varphi_T(1,2)$.

\item $\delta_n$ be the random variable that chooses a tree $T\in \TT_n$ and computes  $\delta_T(1)$.
\end{itemize}
Let us compute now $E(D_n^2)$ from its very definition:
$$
\begin{array}{l}
\displaystyle E(D_n^2)=\sum_{(T,T')\in \TT_n^2} d_{\varphi,2}(T,T')^2p(T)p(T')\\
\quad \displaystyle =\sum_{(T,T')\in \TT_n^2} \Big(\sum_{1\leq i\leq j\leq n} (\varphi_T(i,j)-\varphi_{T'}(i,j))^2\Big)p(T)p(T')\\
\quad \displaystyle 
=\sum_{1\leq i\leq j\leq n}\sum_{(T,T')\in \TT_n^2} (\varphi_T(i,j)^2+\varphi_{T'}(i,j)^2-2\varphi_T(i,j)\varphi_{T'}(i,j))p(T)p(T')
\\
\quad \displaystyle 
=\sum_{1\leq i\leq j\leq n}\hspace*{-1ex}\Big(\sum_{(T,T')\in \TT_n^2}\varphi_T(i,j)^2 p(T)p(T')+\hspace*{-2ex}\sum_{(T,T')\in \TT_n^2} \varphi_{T'}(i,j)^2p(T)p(T')\\
\qquad\qquad \displaystyle 
-2\hspace*{-2ex}\sum_{(T,T')\in \TT_n^2} \varphi_T(i,j)\varphi_{T'}(i,j)p(T)p(T')\Big)\\
\end{array}
$$
$$
\begin{array}{l}
\quad \displaystyle 
=\sum_{1\leq i\leq j\leq n}\Big(\sum_{T\in \TT_n}\varphi_T(i,j)^2 p(T)+\sum_{T'\in \TT_n} \varphi_{T'}(i,j)^2p(T')\\
\qquad\qquad \displaystyle-2\Big(\sum_{T\in \TT_n}\varphi_T(i,j) p(T)\Big)\Big(\sum_{T'\in \TT_n} \varphi_{T'}(i,j) p(T')\Big)\Big)\\
\quad \displaystyle 
=\sum_{1\leq i\leq j\leq n}\Big(2\sum_{T\in \TT_n}\varphi_T(i,j)^2 p(T)
-2\Big(\sum_{T\in \TT_n}\varphi_T(i,j) p(T)\Big)^2\Big)\\
\quad \displaystyle 
=2\sum_{T\in \TT_n} \Big(\sum_{1\leq i\leq j\leq n} \varphi_T(i,j)^2\Big) p(T)-
2\sum_{1\leq i< j\leq n}
\Big(\sum_{T\in \TT_n}\varphi_T(i,j) p(T)\Big)^2\\
\qquad\qquad \displaystyle-2\sum_{1\leq i\leq n}
\Big(\sum_{T\in \TT_n}\varphi_T(i,i) p(T)\Big)^2\\
\quad \displaystyle 
=2\sum_{T\in \TT_n} \oFd(T) p(T)-
2 \binom{n}{2}\Big(\sum_{T\in \TT_n}{\varphi}_T(1,2) p(T)\Big)^2\\
\qquad\qquad \displaystyle-2
n\Big(\sum_{T\in \TT_n}\delta_T(1) p(T)\Big)^2\\
\quad \displaystyle =2E(\oFd_n)-
n(n-1) E(\varphi_n)^2 -2nE(\delta_n)^2
\end{array}
$$
Now, the values of $E(\delta_n)$ and $E(\varphi_n)$ can be easily obtained from $E(S_n)$ and $E(\Phi_n)$, respectively,
using the invariance under relabelings of the probability distribution under which we compute the expected values $E$:
$$
\textstyle E(\delta_n)={E(S_n)}/{n}, \qquad E(\varphi_n)={E(\Phi_n)}/{\binom{n}{2}}
$$
The formula in the statement is then obtained by replacing $E(\delta_n)$ and $E(\varphi_n)$ by these values.
\end{proof}

The expected values of $S_n$ and $\Phi_n$ under the Yule and the uniform models are known:
$$
\begin{array}{ll}
\displaystyle E_Y(S_n)=2n(H_n-1) &\qquad \displaystyle  E_U(S_n)=n\Big(\frac{(2n-2)!!}{(2n-3)!!}-1\Big) \\[2ex]
\displaystyle 
 E_Y(\Phi_n)=n(n-1)-2n(H_n-1) &\qquad
\displaystyle E_U(\Phi_n)=\frac{1}{2} \binom{n}{2} \Big(\frac{(2n-2)!!}{(2n-3)!!}-2 \Big) 
\end{array}
$$
The formula for $E_Y(S_n)$ was proved in \cite{Heard92} and the other three, in \cite{MRR}.

To obtain the expected values of $D_n^2$, it remains to compute the expected values of $\oF_n^{(2)}$. They are given by the following result.

\begin{proposition}\label{prop:D^2}
For every $n\geq 2$,
\begin{enumerate}[(a)]
\item $E_Y(\oF_n^{(2)})=5n(n-1)-8n(H_n-1)$

\item $\displaystyle E_U(\oF_n^{(2)}) = \frac{1}{6} n (4 n^2+21 n-7)-\frac{3}{4}n(n+3)\frac{(2n-2)!!}{(2n-3)!!}$
\end{enumerate}
\end{proposition}

This proposition is proved in the Appendix at the end of this paper. Finally, the identities given in
Theorems  \ref{th:Yule} and  \ref{th:unif}  are obtained by replacing, in the identity given in Proposition \ref{prop:E},
$E(S_n)$, $E(\Phi_n)$, and $E(\oFd_n)$ by their values under the Yule and the uniform models, respectively.
We leave the last details to the reader.

\section{An experiment on TreeBASE}

In this section we report on a very simple experiment to show how $d_{\varphi,2}$ can be used to test evolutionary hypotheses. In this experiment, we have compared the expected value of $d^2_{\varphi,2}$ on $\TT_n$ under the uniform and the Yule models with its average value on the set TreeBASE$_{bin,n}$ of binary phylogenetic trees with $n$ leaves contained in TreeBASE \cite{treebase}.

To perform this experiment, we have taken some decisions. First, since there are only very few values $n> 50$ such that $|\mbox{TreeBASE}_{bin,n}|>10$, we have decided to consider only those binary trees contained in TreeBASE with $n\leq 50$ leaves. On the other hand, even for those $n$ such that $\mbox{TreeBASE}_{bin,n}$ is relatively large,  in most cases it does not contain many pairs of trees  with the same taxa. So, instead of computing the average value of $d^2_{\varphi,2}$ on $\mbox{TreeBASE}_{bin,n}$ by averaging the values $d^2_{\varphi,2}(T,T')$ for pairs of trees $T,T'$ with exactly the same $n$ taxa,
we have made use of the formula given in Proposition \ref{prop:E},
as if $\mbox{TreeBASE}_{bin,n}$ was closed under relabelings: that is, we have taken only into account the shapes of the trees contained in it. 
This is consistent with the fact that our final goal is to test models of evolution that produce tree shapes.

So, we have computed the average values of $\oFd$, of the Sackin index $S$, and of the total cophenetic index $\Phi$ on $\mbox{TreeBASE}_{bin,n}$, and we have taken as average value of $d^2_{\varphi,2}$ on this set the result of appying the formula  in Proposition \ref{prop:E}.  The detailed results of these computations, as well as the Python and R scripts used to compute and analyze them, are available in the Supplementary Material web page  \url{http://bioinfo.uib.es/~recerca/phylotrees/expectedcophdist/}.

Fig. 2 plots the log of these average values as a function of  $\log(n)$. We have added the curves of the log of the expected values of $D_n^2$ under the Yule distribution (lower, dotted curve) and under the uniform distribution (upper, dashed curve), again as a function of  $\log(n)$. The graphic shows that the expected value of $d^2_{\varphi,2}$ on (the shapes of) the phylogenetic trees contained in TreeBASE is better explained by the uniform model than by the Yule  model.
This agrees with the results of similar experiments using other measures (see, for instance, \cite{BF06,MRR}). 

\begin{figure}[htb]
\centering
\includegraphics[height=6cm]{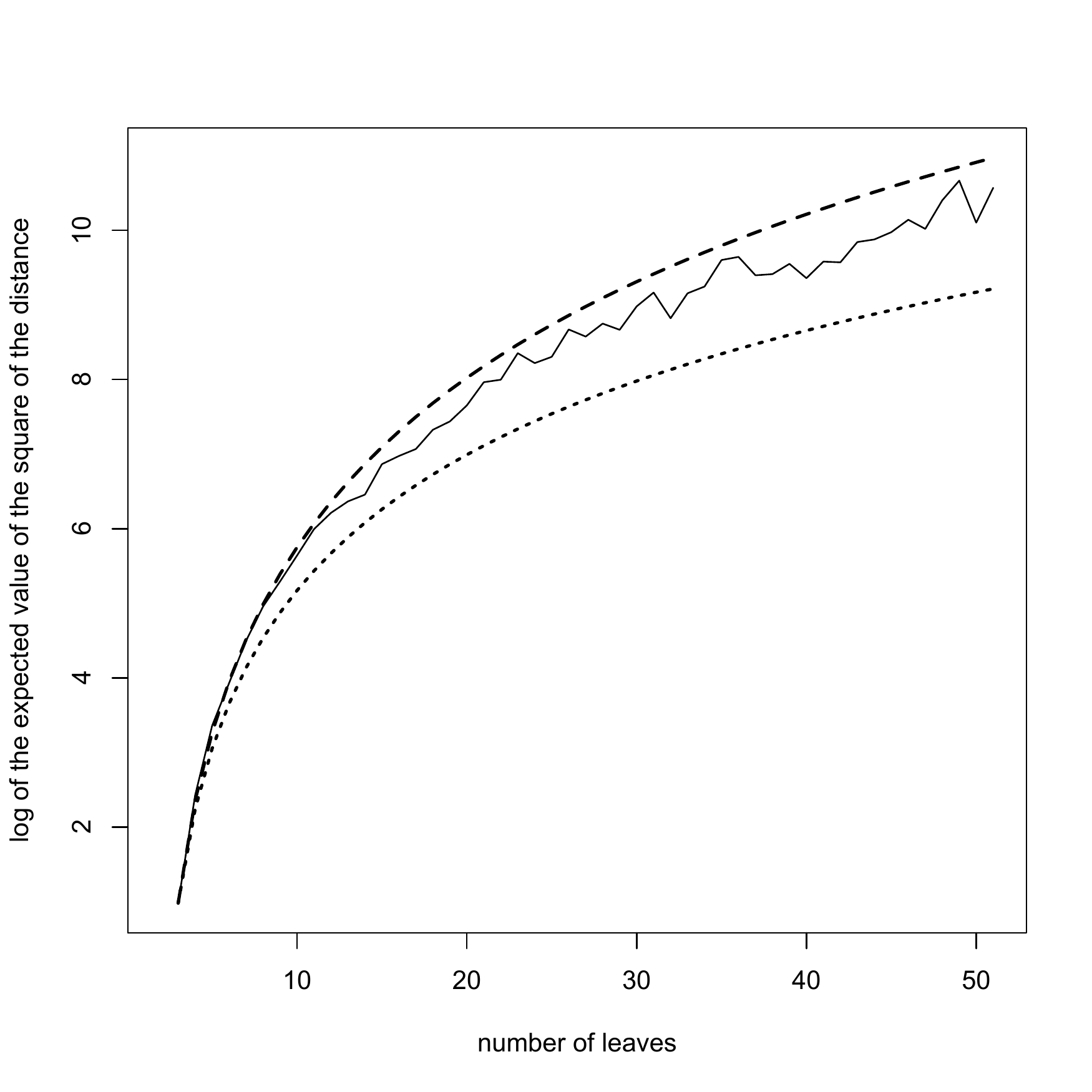}
\label{means}%
\caption{Log-log plots of the mean of $D_n^2$ for the binary trees in TreeBASE with a fixed number  $n$ of leaves,  of $E_Y(D_n^2)$ (dotted curve) and $E_U(D_n^2)$ (dashed curve).}
\end{figure}%

\section{Conclusions and discussion}

In this paper we have obtained formulas for the expected values under the Yule and the uniform models of the square of the euclidean cophenetic metric $d_{\varphi,2}$, defined by the euclidean distance between cophenetic vectors.  These formulas are explicit  and hold on spaces $\TT_n$ of fully resolved phylogenetic trees with any number $n$ of leaves.

These formulas have been obtained through long algebraic manipulations of sums of sequences. To double-check our results, we have computed the exact value of $E_Y(D_n^2)$ and $E_U(D_n^2)$ for $n=3,\dots,7$, by generating all trees with up to $7$ leaves. Moreover, we have computed numerical approximations to these values  for $n=10,20,\dots,100$, by generating pairs of random trees until the numerical method stabilizes. These numerical experiments confirm that our formulas give the right figures. Table \ref{table:1} gives the exact values for  $n=3,\dots,7$.  The results of the simulations for $n=10,20,\dots,100$, as well as
the Python scripts used in these computations,  are also available in the aforementioned Supplementary Material web page  \url{http://bioinfo.uib.es/~recerca/phylotrees/expectedcophdist/}.

\begin{center}
\begin{table}[htb]
\begin{center}
\begin{tabular}{r|ccccc}
& 3& 4 & 5 & 6 & 7 \\
\hline
$E_Y(D_n^2)$ &
2.66667 & 9.40741 & 21.1833 & 38.712 & 62.5562\\ 
$E_U(D_n^2)$ & 2.66667 & 10.56 & 26.2367 & 52.3023 & 91.4086
\end{tabular}
\end{center}
\caption{\label{table:1}Values of $E_Y(D_n^2)$ and $E_U(D_n^2)$ for $n=3,\ldots,7$. They agree with those given by our formulas.}
\end{table}
\end{center}

The  formulas for $E_Y(D_n^2)$ and $E_U(D_n^2)$ grow in different orders: $E_Y(D_n^2)$ is in $\Theta(n^2)$, while $E_U(D_n^2)$ is in $\Theta(n^3)$. Therefore, they can be used to test the Yule and the uniform models as null stochastic models of evolution for collections of phylogenetic trees reconstructed by different methods. 
We have reported on a first experiment of this type, which reinforces the conclusion that ``real world'' phylogenetic trees (that is, those contained in TreeBASE)  are not consistent with the Yule model of evolution. We plan to report  in a future paper on more extensive tests on stochastic models of evolutionary processes, including Ford's $\alpha$-model \cite{Ford} and Aldous' $\beta$-model \cite{Ald1}.

\section*{Acknowledgements}  The research reported in this paper has been partially supported by the Spanish government and the UE FEDER program, through projects MTM2009-07165 and TIN2011-15874-E.


\section*{Appendix: Proof of Proposition 5}
\subsection*{Proof of Proposition \ref{prop:D^2}.(a)}

For every $T\in \TT_n$, let
$$
\oF(T)=S(T)+\Phi(T)=\sum_{1\leq i\leq j\leq n}\varphi_T(i,j), 
$$
and let $\oF_n$ be the random variable that chooses a tree $T\in \TT_n$ and computes $\oF(T)$. We have that
$$
E_Y(\oF_n)=E_Y(S_n)+E_Y(\Phi_n)=n(n-1).
$$

To compute $E_Y(\oFd_n)$, we shall use an argument similar to the one used in the proof of \cite[Prop. 3]{Yulenodal}. Notice that
$$
\begin{array}{l}
E_Y(\oF_n^{(2)}) \displaystyle =\sum_{T\in \TT_n} \oF^{(2)}(T)\cdot p_Y(T)
\\
\quad \displaystyle =\frac{1}{2}
\sum_{k=1}^{n-1}\sum_{S_k\subsetneq\{1,\ldots,n\}\atop |S_k|=k}
 \sum_{T_k\in \TT(S_k)}\sum_{T'_{n-k}\in \TT(S_k^c)}\oF^{(2)}(T_k\widehat{\ }\, {}T'_{n-k})\cdot p_Y(T_k\widehat{\ }\, {}T'_{n-k})
\end{array}
$$
Now, on the one hand,  we have the following easy lemma on $P_Y(T\,\widehat{\ }\, T')$: see  \cite[Lem. 1]{CMR}.

\begin{lemma}\label{lem:pangle}
Let $\emptyset\neq S_k\subsetneq \{1,\ldots,n\}$ with $|S_k|=k$, let  $T_k\in \TT(S_k)$ and $T'_{n-k}\in \TT( S_k^c)$. Then,
$$
P_Y(T_k\widehat{\ }\, {}T'_{n-k})=\dfrac{2}{(n-1)\binom{n}{k}} P(T_{k})P(T'_{n-k}).
$$
\end{lemma}

On the other hand, we have the following recursive expression for  $\oFd(T\,\widehat{\ }\, T')$.

\begin{lemma}
Let $\emptyset\neq S_k\subsetneq \{1,\ldots,n\}$ with $|S_k|=k$, let  $T_k\in \TT(S_k)$ and $T'_{n-k}\in \TT( S_k^c)$.
Then
$$
\oFd(T_k\,\widehat{\ }\, T_{n-k}')=\oFd(T_k)+\oFd(T_{n-k}')+2\oF(T_k)+2\oF(T_{n-k}')+\binom{k+1}{2}+\binom{n-k+1}{2}.
$$
\end{lemma}

\begin{proof}
Let us assume, without any loss of generality, that $S=\{1,\ldots,m\}$ and $S'=\{m+1,\ldots,n\}$. Then
$$
\varphi_{T_k\,\widehat{\ }\, T_{n-k}'}(i,j)=\left\{
\begin{array}{ll}
\varphi_{T_k}(i,j)+1 & \mbox{ if $1\leq i,j\leq k$}\\
\varphi_{T_{n-k}'}(i,j)+1 & \mbox{ if $k+1\leq i,j\leq n$}\\
0 & \mbox{ otherwise}
\end{array}
\right.
$$
and therefore
$$
\begin{array}{l}
\displaystyle 
\oFd(T_k\,\widehat{\ }\, T_{n-k}')=\sum_{1\leq i\leq j\leq n} \varphi_{T_k\,\widehat{\ }\, T_{n-k}'}(i,j)^2\\
\quad \displaystyle 
=
\sum_{1\leq i\leq j\leq k} (\varphi_{T_k}(i,j)+1)^2+\sum_{k+1\leq i\leq j\leq n} (\varphi_{T_{n-k}'}(i,j)+1)^2\\
\quad \displaystyle 
=\sum_{1\leq i\leq j\leq k} (\varphi_{T_k}(i,j)^2+2\varphi_{T_k}(i,j)+1)+\hspace*{-3ex}\sum_{k+1\leq i\leq j\leq n} (\varphi_{T_{n-k}'}(i,j)^2+2\varphi_{T_{n-k}'}(i,j)+1)\\
\quad \displaystyle 
=\oFd(T_k)+2\oF(T_{k})+\binom{k+1}{2}+\oFd(T_{n-k}')+2\oF(T_{n-k}')+\binom{n-k+1}{2}.
\end{array}
$$
\end{proof}

So, if we set
$f(a,b)=\binom{a+1}{2}+\binom{b+1}{2}$,
we have that 
$$
\begin{array}{l}
E_Y(\oF_n^{(2)})
\\
\quad   \displaystyle =\frac{1}{2} \sum_{k=1}^{n-1}\binom{n}{k}
 \sum_{T_k\in \TT_k}\sum_{T'_{n-k} \in \TT_{n-k}}\hspace*{-2ex}\Big[\oFd(T_k)+\oFd(T'_{n-k})+2(\oF(T_k)+\oF(T_{n-k}'))\\
 \displaystyle  \qquad+f(k,n-k)\Big]
 \frac{2}{(n-1)\binom{n}{k}} P_Y(T_{k})P_Y(T'_{n-k})\\
 \displaystyle\quad   =\frac{1}{n-1}\sum_{k=1}^{n-1} 
\Big[ \sum_{T_k}\sum_{T'_{n-k} }\oFd(T_k)P_Y(T_{k})P_Y(T'_{n-k})   \\
\displaystyle\qquad +\sum_{T_k}\sum_{T'_{n-k} }\oFd(T'_{n-k})P_Y(T_{k})P_Y(T'_{n-k}) \\
\displaystyle\qquad +2\sum_{T_k}\sum_{T'_{n-k} }\oF(T_k)P_Y(T_{k})P_Y(T'_{n-k})\\
\displaystyle\qquad +2\sum_{T_k}\sum_{T'_{n-k} }\oF(T'_{n-k})P_Y(T_{k})P_Y(T'_{n-k}) \\
\displaystyle\qquad +\sum_{T_k}\sum_{T'_{n-k} }f(k,n-k)P_Y(T_{k})P_Y(T'_{n-k}) \Big]\\
  \quad  \displaystyle  =\frac{1}{n-1}\sum_{k=1}^{n-1} 
\Big[ \sum_{T_k}\oFd(T_k)P_Y(T_{k})  +\sum_{T'_{n-k} }\oFd(T'_{n-k})P_Y(T'_{n-k})\\
\displaystyle\qquad +2\sum_{T_k}\oF(T_k)P_Y(T_{k}) +2\sum_{T'_{n-k} }\oF(T'_{n-k})P_Y(T'_{n-k}) +f(k,n-k) \Big]\\
\quad  \displaystyle  =\frac{1}{n-1}\sum_{k=1}^{n-1} 
\Big[E_Y(\oFd_k)+E_Y(\oFd_{n-k}) + 2E_Y(\oF_k)+2E_Y(\oF_{n-k})\\
\displaystyle\qquad +\binom{k+1}{2}+\binom{n-k+1}{2}\Big]\\
\quad  \displaystyle  =\frac{2}{n-1}\sum_{k=1}^{n-1} 
E_Y(\oFd_k) + \frac{4}{n-1}\sum_{k=1}^{n-1}E_Y(\oF_k)+\frac{1}{3}n(n+1).
\end{array}
$$
In particular
$$
 E_Y(\oF_{n-1}^2) =\frac{2}{n-2}\sum_{k=1}^{n-2} 
E_Y(\oFd_k) + \frac{4}{n-2}\sum_{k=1}^{n-2}E_Y(\oF_k)+\frac{1}{3}n(n-1).$$
and therefore
$$
\begin{array}{l}
\displaystyle  E_Y(\oFd_n)=
\frac{n-2}{n-1}\cdot \frac{2}{n-2}\sum_{k=1}^{n-2} E_Y(\oFd_k) +\frac{2}{n-1}E_Y(\oFd_{n-1})\\
\qquad\qquad\quad \displaystyle +
\frac{n-2}{n-1}\cdot \frac{4}{n-2}\sum_{k=1}^{n-2} E_Y(\oF_k) +\frac{4}{n-1}E_Y(\oF_{n-1})\\
\qquad\qquad\quad \displaystyle+\frac{n-2}{n-1}\cdot\frac{1}{3}n(n-1)+ n\\[1ex]
\quad \displaystyle=\frac{n-2}{n-1} E_Y(\oFd_{n-1})+\frac{2}{n-1}E_Y(\oFd_{n-1})+\frac{4}{n-1}E_Y(\oF_{n-1})+n\\[1ex]
\quad \displaystyle=\frac{n}{n-1} E_Y(\oFd_{n-1})+5n-8.
\end{array}
$$
Setting $x_n=E_Y(\oFd_n)/n$, this recurrence becomes
$$
x_n=x_{n-1}+5-\frac{8}{n}
$$
and the solution of this recursive equation with $x_1=E_Y(\oFd_1)=0$ is
$$
x_n=\sum_{k=2}^n\Big(5-\frac{8}{k}\Big) =5(n-1)-8(H_n-1)=5n+3-8H_n
$$
from where we deduce that $E_Y(\oFd_n)=5n^2+3n-8nH_n$,
as we claimed.

\subsection*{Proof of Proposition \ref{prop:D^2}.(b)}

To compute $E_U(\oFd_n)$, we shall use an argument similar to the one used in \cite{MirR10}. For every $k=1,\ldots,n-1$, let  
$$
\begin{array}{rl}
f_{k,n}& =|\{ T\in \TT_n\mid \varphi_T(1,2)=k\}|\\ & =|\{ T\in \TT_n\mid \varphi_T(i,j)=k\}| \mbox{  for every $1\leq i<j\leq n$}\\
d_{k,n}&=|\{ T\in \TT_n\mid \delta_T(1)=k\}|\\ & =|\{ T\in \TT_n\mid \delta_T(i)=k\}| \mbox{  for every $1\leq i\leq n$}
\end{array}
$$
(where $|X|$ denotes the cardinal of the set $X$).

\begin{lemma}\label{lem:9}
For every $n\geq 2$,
$$
E_U(\oFd_n)=\frac{1}{(2n-3)!!}\Big(n\sum_{k=1}^{n-1} k^2\cdot d_{k,n}+\binom{n}{2}\sum_{k=1}^{n-2} k^2\cdot f_{k,n}\Big)
$$
\end{lemma}

\begin{proof}
Under the uniform model,
$$
E_U(\oFd_n)=\frac{\sum_{T\in \TT_n} \oFd(T)}{(2n-3)!!},
$$
where
$$
\begin{array}{l}
\displaystyle \sum_{T\in \TT_n} \oFd(T)   =\sum_{T\in \TT_n}\sum_{1 \leq i\leq j\leq n} \varphi_T(i,j)^2 =\sum_{1 \leq i\leq j\leq n}\sum_{T\in \TT_n} \varphi_T(i,j)^2\\
\quad \displaystyle =\sum_{1 \leq i\leq n}\sum_{T\in \TT_n} \delta_T(i)^2   +\sum_{1 \leq i< j\leq n}\sum_{T\in \TT_n} \varphi_T(i,j)^2\\
\quad \displaystyle =\sum_{1 \leq i\leq n}\sum_{k=1}^{n-1} k^2\cdot |\{ T\in \TT_n\mid \delta_T(i)=k\}|\\
\qquad\qquad \displaystyle  +\sum_{1 \leq i< j\leq n}\sum_{k=1}^{n-2} k^2\cdot |\{ T\in \TT_n\mid \varphi_T(i,j)=k\}|\\
\quad \displaystyle =\sum_{1 \leq i\leq n}\sum_{k=1}^{n-1} k^2\cdot d_{k,n}+\sum_{1 \leq i< j\leq n}\sum_{k=1}^{n-2} k^2\cdot f_{k,n}\\
\quad \displaystyle =n\sum_{k=1}^{n-1} k^2\cdot d_{k,n}+\binom{n}{2}\sum_{k=1}^{n-2} k^2\cdot f_{k,n}.
\end{array}
$$
\end{proof}

A formula for $d_{k,n}$ was obtained in the proof of \cite[Lem. 21]{MRR}:
\begin{equation}\label{(d)}
d_{k,n}= \frac{(2n-k-3)!\cdot k}{(n-k-1)!2^{n-k-1}}.
\end{equation}
As far as $f_{k,n}$ goes, we have the following result. In it, and henceforth, ${}_pF_q$
denotes  the (\emph{generalized}) \emph{hypergeometric function} defined byå
$$
{}_pF_q\left(\begin{array}{rrr} a_1,&\ldots, &a_p \\[-0.5ex] b_1,& \ldots,  &b_q\end{array};z\right)=\sum_{k\geq 0} \frac{(a_1)_k\cdots (a_p)_k}{(b_1)_k\cdots (b_q)_k}\cdot \frac{z^k}{k!},
$$ 
where $(a)_0=1$ and $(a)_k := a\cdot (a+1)\cdots (a+k-1)$ for $k\geq 1$.

\begin{lemma}\label{lem:f}
For every $n\geq 2$, $f_{0,n}=(2n-4)!!$ and
$$
f_{k,n}=\frac{(2n-k-5)!k}{(2n-2k-4)!!} \cdot {}_3 F_2
\bigg(\begin{array}{l} 1,\ 2-n,\ k+2-n \\[-0.5ex]  \frac{k+5}{2}-n,\  \frac{k}{2}-n+3\end{array};1\bigg)
$$
for every $k=1,\ldots,n-2$.
\end{lemma}

\begin{proof}
Let us start by proving $f_{0,n}=(2n-4)!!$ by induction on $n$. It is clear that $f_{0,2}=1=(2\cdot 2-4)!!$. Assume now that $f_{0,n-1}=(2(n-1)-4)!!$.
Every phylogenetic tree $T$ with $n$ leaves such that $\varphi_T(1,2)=0$, that is, where $LCA_T(1,2)$ is the root, is obtained by taking a phylogenetic tree $T'$ with $n-1$ leaves such that $\varphi_{T'}(1,2)=0$ and adding a new pendant edge, ending in the leaf $n$, to any edge in $T'$.
Then, since there are $f_{0,n-1}=(2n-6)!!$ trees $T'\in \TT_{n-1}$ such that $\varphi_{T'}(1,2)=0$, and each one of them has $2(n-1)-2$ edges where we can add the new edge, we obtain
$$
f_{0,n}=(2n-4)(2n-6)!!=(2n-4)!!.
$$

Now, to compute $f_{k,n}$ for $k\geq 1$, we shall study the structure of a tree
$T\in \TT_n $ such that $\varphi_T(1,2)=k$; to simplify the notations, let us denote by $x$ the node $LCA_T(1,2)$, which has depth $k$, and by $T_0$ the subtree of $T$ rooted at $x$.

Then, on the one hand, $T_0$ is a phylogenetic tree on a subset $S_0\subseteq \{1,\ldots,n\}$ containing $1,2$, and since its root $x$ is the LCA of 1 and 2 in $T$, we have that $\varphi_{T_0}(1,2)=0$. On the other hand, there is a path $(r=v_1,v_2,v_3,\ldots,v_{k+1}=x)$ in $T$ from $r$ to $x$. For every $j=1,\ldots, k$, let $T_j$ be the subtree rooted at the child of $v_j$ other than $v_{j+1}$;  see Fig.~\ref{fig:forests}.

So, the tree $T$ is determined by:
\begin{itemize}
\item A number $0\leq m\leq n-k-2$, so that $m+2$ will be the number of leaves of the phylogenetic tree $T_0$ rooted at $LCA_T(1,2)$
  
\item A subset $\{i_1,\ldots,i_m\}$ of $\{3,\ldots, n\}$. There are $\binom{n-2}{m}$  such subsets.

\item A phylogenetic tree $T_0$ on $\{1,2,i_1,\ldots,i_m\}$ such that $\varphi_{T_0}(1,2)=0$. There are $f_{0,m+2}=(2m)!!$ such trees.

\item An \emph{ordered $k$-forest}, that is, an ordered sequence of phylogenetic trees $(T_1,T_1,\ldots,T_{k})$ such that $\bigcup_{i=1}^k L(T_i)=\{1,\ldots,n\}-\{1,2,i_1,\ldots,i_m\}$. The number of such ordered $k$-forests is (see, for instance,
\cite[Lem.~1]{MirR10})
$$
\frac{(2n-2m-k-5)!k}{(n-m-k-2)!2^{n-m-k-2}}.
$$ 
\end{itemize}

\begin{figure}[htb]
\begin{center}
\begin{tikzpicture}[thick,>=stealth,scale=0.3]
\draw(0,2) node[tre] (x) {}; \etq x
\draw(-2,-1.5) node  {\scriptsize $1$}; 
\draw(2,-1.5) node  {\scriptsize $2$}; 
\draw (x)--(-2,-1)--(2,-1)--(x);
\draw(0,0) node  {\footnotesize $T_0$};
\draw(2,4) node[tre] (x1) {}; 
\draw (x1)--(x);
\draw (x1)--(3,5);
\draw (3.3,5.3) node {.};
\draw (3.6,5.6) node {.};
\draw (3.9,5.9) node {.};
\draw(5,7) node[tre] (x2) {}; 
\draw (4.2,6.2)--(x2);
\draw(7,9) node[tre] (r) {}; 
\draw (x2)--(r);
\draw(5,3) node[tre] (tk) {}; 
\draw (tk)--(3.6,0)--(6.4,0)--(tk);
\draw(5,1) node  {\footnotesize $T_{k}$};
\draw (x1)--(tk);
\draw(8,6) node[tre] (t2) {}; 
\draw (t2)--(6.6,3)--(9.4,3)--(t2);
\draw(8,4) node  {\footnotesize $T_{2}$};
\draw (x2)--(t2);
\draw(10,8) node[tre] (t1) {}; 
\draw (t1)--(8.6,5)--(11.4,5)--(t1);
\draw(10,6) node  {\footnotesize $T_{1}$};
\draw (r)--(t1);
\end{tikzpicture}
\end{center}
\caption{\label{fig:forests} 
The structure of a tree $T$ with $\varphi_T(1,2)=k$.}
\end{figure}
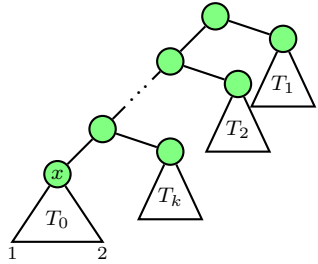

This shows that $f_{k,n}$ can be computed as 
$$
\begin{array}{rl}
f_{k,n} & = \displaystyle\sum _{m=0}^{n-k-2}  \mbox{(number of ways  of choosing $\{i_1,\ldots,i_m\}$)}\\[-2ex] &
 \qquad\qquad \cdot \mbox{(number of trees in $\TT_{m+2}$ with $\varphi_T(1,2)=0$)} \\
& \qquad\qquad \cdot \mbox{(number of  ordered $k$-forests  on $n-m-2$ leaves)} \\
& =
\displaystyle \sum_{m=0}^{n-k-2}  \binom{n-2}{m}\cdot (2m)!! \cdot \frac{(2n-2m-k-5)!k}{(n-m-k-2)!2^{n-m-k-2}}\\
& = \displaystyle k\sum_{m=0}^{n-k-2} \frac{(n-2)! m! 2^m (2n-2m-k-5)!}{m!(n-m-2)!(n-m-k-2)!2^{n-m-k-2}}\\
& = \displaystyle\frac{(n-2)!k}{2^{n-k-2}}\sum_{m=0}^{n-k-2} \frac{4^m (2n-2m-k-5)!}{(n-m-2)!(n-m-k-2)!}
\end{array}
$$
Now, taking into account that
$$
\begin{array}{l}
(1)_m =  m!\\
(2-n)_m =(-1)^m\dfrac{(n-2)!}{(n-m-2)!}\\
(k+2-n)_m =  (-1)^m\dfrac{(n-k-2)!}{(n-k-m-2)!} \\
 \left(\dfrac{k+5}{2}-n\right)_m =  \dfrac{(-1)^m (2n-k-5)!!}{2^m (2n-k-2m-5)!!},\\
\left(\dfrac{k}{2}-n+3\right)_m  = \dfrac{(-1)^m (2n-k-6)!!}{2^m (2n-k-2m-6)!!}
\end{array}
$$
we have that
$$
\begin{array}{l}
\displaystyle {}_3 F_2\bigg(\begin{array}{l} 1,\ 2-n,\ k+2-n \\ \frac{k+5}{2}-n,\  \frac{k}{2}-n+3\end{array};1\bigg)
=\sum_{m\geq 0} \frac{(1)_m\cdot (2-n)_m\cdot (k+2-n)_m}{(\frac{k+5}{2}-n)_m\cdot (\frac{k}{2}-n+3)_m} \cdot \frac{1}{m!}\\
\quad\displaystyle =\sum_{m\geq 0}\frac{m!(n-2)! (n-k-2)!2^m(2n-k-2m-5)!!2^m(2n-k-2m-6)!!}{(n-m-2)! (n-k-m-2)!(2n-k-5)!!(2n-k-6)!!m!}\\
\quad\displaystyle =\sum_{m= 0}^{n-k-2}\frac{(n-2)! (n-k-2)!(2n-k-2m-5)!2^{2m}}{(n-m-2)! (n-k-m-2)!(2n-k-5)!}\\
\quad\displaystyle =\frac{(n-2)! (n-k-2)!}{(2n-k-5)!}\sum_{m= 0}^{n-k-2}\frac{(2n-k-2m-5)!4^{m}}{(n-m-2)! (n-k-m-2)!}
\end{array}
$$
from where we deduce that
$$
\begin{array}{l}
\displaystyle \sum_{m= 0}^{n-k-2}\frac{(2n-k-2m-5)!4^{m}}{(n-m-2)! (n-k-m-2)!}\\
\qquad\qquad\displaystyle=\frac{(2n-k-5)!}{(n-2)! (n-k-2)!}{}_3 F_2\bigg(\begin{array}{l} 1,\ 2-n,\ k+2-n \\ \frac{k+5}{2}-n,\  \frac{k}{2}-n+3\end{array};1\bigg)
\end{array}
$$
and hence
$$
\begin{array}{rl}
f_{k,n} & = \displaystyle \displaystyle\frac{(n-2)!k}{2^{n-k-2}}\sum_{m=0}^{n-k-2} \frac{4^m (2n-2m-k-5)!}{(n-m-2)!(n-m-k-2)!}\\
& \displaystyle =\frac{(n-2)!k}{2^{n-k-2}} \cdot \frac{(2n-k-5)!}{(n-2)! (n-k-2)!}{}_3 F_2\bigg(\begin{array}{l} 1,\ 2-n,\ k+2-n \\ \frac{k+5}{2}-n,\  \frac{k}{2}-n+3\end{array};1\bigg)\\
& \displaystyle =\frac{(2n-k-5)!k}{(2n-2k-4)!!}\cdot {}_3 F_2\bigg(\begin{array}{l} 1,\ 2-n,\ k+2-n \\ \frac{k+5}{2}-n,\  \frac{k}{2}-n+3\end{array};1\bigg)
\end{array}
$$
as we claimed.
\end{proof}

We must compute now the sums
$$
\sum_{k=1}^{n-1} k^2\cdot d_{k,n},\quad \sum_{k=1}^{n-2} k^2\cdot f_{k,n}.
$$
To do that, we shall use the following auxiliary lemma.

\begin{lemma}\label{lem:U}
For every $n\geq 2$ and $m\geq 1$, let
$$
U_{n,m} = \sum_{k=0}^{n-2} \frac{k^m (n+k-2)! }{k!2^{k}}.
$$
Then,
$$
\begin{array}{rl}
U_{n,0} &=  (2n-4)!! \\
U_{n,1} &=  (n-1)(2n-4)!!- (2n-3)!! \\
U_{n,2} &=  (n^2-1)(2n-4)!!-(2n-1)(2n-3)!! \\
U_{n,3} &= (n^3+3n^2-3n-1)(2n-4)!! -(3n^2+n-1) (2n-3)!!
\end{array}
$$
\end{lemma}

\begin{proof}
The proof of these identities is standard, using well known equalities for hypergeometric functions and the \emph{lookup algorithm} given in \cite[p. 36]{AeqB}. We shall prove in detail the identity for $m=2$, and we leave the details of the rest to the reader.

Notice that
$$
\begin{array}{l}
\displaystyle U_{n,2}=\sum_{k=0}^{n-2} \frac{k^2 (n+k-2)! }{k!2^{k}}=\sum_{k=1}^{n-2} \frac{k^2 (n+k-2)! }{k!2^{k}}=\sum_{k=0}^{n-3} \frac{(k+1)^2 (n+k-1)! }{(k+1)!2^{k+1}}\\
\displaystyle \qquad
=\sum_{k=0}^{\infty} \frac{(k+1)^2 (n+k-1)! }{(k+1)!2^{k+1}}
-\sum_{k=n-2}^{\infty} \frac{(k+1)^2 (n+k-1)! }{(k+1)!2^{k+1}}
\end{array}
$$
Set
$$
X_n=\sum_{k=0}^{\infty} \frac{(k+1)^2 (n+k-1)! }{(k+1)!2^{k+1}},\qquad Y_n=\sum_{k=n-2}^{\infty} \frac{(k+1)^2 (n+k-1)! }{(k+1)!2^{k+1}}
$$
We compute now these two summands.

As to $X_n$,
$$
X_n=\frac{(n-1)!}{2}\sum_{k=0}^{\infty} \frac{(k+1)^2 (n+k-1)! }{(n-1)!(k+1)!2^{k}}
$$
If we set
$$
t_{k}=\frac{(k+1)^2 (n+k-1)! }{(n-1)!(k+1)!2^{k}},
$$
we have that
$$
\frac{t_{k+1}}{t_k}=\frac{(k+2)(k+n)}{(k+1)^2}\cdot \frac{1}{2}
$$
and therefore, by the \emph{lookup algorithm} \cite[p. 36]{AeqB}, we have that
$$
\begin{array}{l}
X_n  \displaystyle =\frac{(n-1)!}{2}\cdot {}_2 F_1
\bigg(\begin{array}{l} 2,\ n \\[-0.5ex] 1\end{array};\frac{1}{2}\bigg)\\
\quad \displaystyle =\frac{(n-1)!}{2}\cdot 2^n\cdot {}_2 F_1
\bigg(\begin{array}{l} n,\ -1 \\[-0.5ex] 1\end{array};-1\bigg)\quad\mbox{(using (15.3.4) in \cite[p. 559]{AS})}\\
\quad  \displaystyle =(n-1)!2^{n-1}\sum_{k\geq 0} \frac{(n)_k(-1)_k}{(1)_k}\cdot \frac{(-1)^k}{k!}\\
\quad  \displaystyle =(n-1)!2^{n-1}\Big(\frac{(n)_0(-1)_0}{(1)_0}\cdot \frac{(-1)^0}{0!}+
\frac{(n)_1(-1)_1}{(1)_1}\cdot \frac{(-1)^1}{1!}\Big)\\
\quad  \displaystyle=(n-1)!2^{n-1}(n+1)
\end{array}
$$

As to $Y_n$,
$$
\begin{array}{l} 
\displaystyle Y_n=
\sum_{k=0}^{\infty} \frac{(k+n-1)^2 (2n+k-3)! }{(k+n-1)!2^{k+n-1}}
\\
\qquad \displaystyle  =
\frac{(n-1)^2(2n-3)!}{(n-1)! 2^{n-1}}\cdot 
\sum_{k=0}^{\infty} \frac{(k+n-1)^2 (2n+k-3)! }{(k+n-1)!2^{k}\cdot \frac{(n-1)^2(2n-3)!}{(n-1)!}}
\end{array}
$$
If we take now
$$
t_{k}= \frac{(k+n-1)^2 (2n+k-3)! }{(k+n-1)!2^{k}\cdot \frac{(n-1)^2(2n-3)!}{(n-1)!}}
$$
we have that
$$
\frac{t_{k+1}}{t_k}=\frac{(n+k)(2n+k-2)}{(k+n-1)^2}\cdot \frac{1}{2}
$$
and therefore, again by the \emph{lookup algorithm} \cite[p. 36]{AeqB}, we have that
$$
\begin{array}{l}
Y_n  \displaystyle =\frac{(n-1)^2(2n-3)!}{(n-1)! 2^{n-1}}\cdot {}_3 F_2
\bigg(\begin{array}{l} 1,\ n,\ 2n-2 \\[-0.5ex] n-1,\ n-1\end{array};\frac{1}{2}\bigg)\\
\qquad \displaystyle  =
\frac{(n-1)^2(2n-3)!}{(n-1)! 2^{n-1}}\Big[
{}_2 F_1
\bigg(\begin{array}{l} 2n-2,\ 1 \\[-0.5ex] n-1\end{array};\frac{1}{2}\bigg)+\frac{1}{n-1}
\cdot {}_2 F_1
\bigg(\begin{array}{l} 2n-1,\ 2 \\[-0.5ex] n\end{array};\frac{1}{2}\bigg)\Big]\\
\quad \hspace*{\fill} \mbox{(using \cite{HF1})}.
\end{array}
$$
Now
$$
\begin{array}{l}
\displaystyle {}_2 F_1
\bigg(\begin{array}{l} 2n-2,\ 1 \\[-0.5ex] n-1\end{array};\frac{1}{2}\bigg)=2\cdot
{}_2 F_1
\bigg(\begin{array}{l} 1-n,\ 1 \\[-0.5ex] n-1\end{array};-1\bigg)\quad\mbox{(using (15.3.4) in \cite[p. 559]{AS})}\\
\qquad \displaystyle=2\cdot \frac{2^{2(n-2)}\Gamma(n-1)}{\Gamma(2n-2)}\Big[\frac{\Gamma(n-1)}{\Gamma(0)}+
\frac{\Gamma(n)}{\Gamma(1)}+\frac{2\Gamma(n-\frac{1}{2})}{\Gamma(\frac{1}{2})}\Big]\quad\mbox{(using \cite{HF2})}\\
\qquad \displaystyle =2\cdot \frac{2^{2(n-2)}(n-2)!}{(2n-3)!}\Big[(n-1)!+2\cdot \frac{(2n-3)!!}{2^{n-1}}\Big]\\
\qquad \displaystyle=\frac{2^{n-1}(n-1)!}{(2n-3)!!}+2\\[2ex]
\displaystyle  {}_2 F_1
\bigg(\begin{array}{l} 2n-1,\ 2 \\[-0.5ex] n\end{array};\frac{1}{2}\bigg)=
2^2\cdot
{}_2 F_1
\bigg(\begin{array}{l} 2,\ 1-n \\[-0.5ex] n\end{array};-1\bigg)\quad\mbox{(using (15.3.4) in \cite[p. 559]{AS})}\\
\qquad \displaystyle=
4\cdot \frac{\Gamma(n)}{2^{2(2-n)}\Gamma(2n-1)}\Big(\frac{\Gamma(n-\frac{1}{2})}{\Gamma(\frac{1}{2})}+\frac{\Gamma(n+\frac{1}{2})}{\Gamma(\frac{3}{2})}+2\Gamma(n)\Big) \quad\mbox{(using \cite{HF2})}\\
\qquad \displaystyle=
\frac{2^{2n-2}(n-1)!}{(2n-2)!}\Big(\frac{(2n-3)!!}{2^{n-1}}+\frac{(2n-1)!!}{2^{n-1}}+2\cdot (n-1)!\Big)\\
\qquad \displaystyle=
\frac{2^{n-1}(n-1)!}{(2n-2)!}({(2n-3)!!}+{(2n-1)!!}+2^n\cdot (n-1)!)
\end{array}
$$
Therefore,
$$
\begin{array}{l}
\displaystyle Y_n=\frac{(n-1)^2(2n-3)!}{(n-1)! 2^{n-1}}\Big[\frac{2^{n-1}(n-1)!}{(2n-3)!!}+2\\
\qquad\qquad\qquad \displaystyle+\frac{1}{n-1}\cdot
\frac{2^{n-1}(n-1)!}{(2n-2)!}({(2n-3)!!}+{(2n-1)!!}+2^n\cdot (n-1)!)\Big]\\
\quad\displaystyle =2^{n-2}(n+1)(n-1)!+ (2n-1)!!
\end{array}
$$
and finally
$$
\begin{array}{rl}
U_{n,2}&=X_n-Y_n=2^{n-2}(n+1)(n-1)!-(2n-1)!!\\
&= (n^2-1)(2n-4)!!-(2n-1)(2n-3)!!
\end{array}
$$
as we claimed.
\end{proof}

\begin{lemma}\label{sum:d}
For every $n\geq 2$,
$$
\sum_{k=1}^{n-1}  k^2d_{k,n}= (4n-1) (2n-3)!! -3  (2n-2)!!. 
$$
\end{lemma}

\begin{proof}
By equation (\ref{(d)}),
$$
\begin{array}{l}
\displaystyle \sum_{k=1}^{n-1}  k^2d_{k,n} = \sum_{k=1}^{n-1}  \frac{k^3(2n-k-3)!}{(n-k-1)!2^{n-k-1}}= \sum_{k=0}^{n-2}  \frac{(n-k-1)^3(n+k-2)!}{k!2^{k}}\\
\qquad \displaystyle = (n-1)^3 U_{n,0}-3 (n-1)^2 U_{n,1}+3 (n-1) U_{n,2}-U_{n,3}\\
\qquad =(n-1)^3 (2n-4)!!-3 (n-1)^2\big( (n-1)(2n-4)!!- (2n-3)!!\big)\\
\qquad \qquad
+3 (n-1)\big( (n^2-1)(2n-4)!!-(2n-1)(2n-3)!!\big)\\
\qquad\qquad -\big((n^3+3n^2-3n-1)(2n-4)!! -(3n^2+n-1) (2n-3)!!\big)\\
\qquad =(4n-1)(2n-3)!!-3(2n-2)(2n-4)!!.
\end{array}
$$
\end{proof}

\begin{lemma}\label{sum:f}
For every $n\geq 2$,
$$
 \sum_{k=1}^{n-2} k^2 f_{k,n} = \frac{1}{3} (4n+1) (2n-3)!!-\frac{3}{2} (2n-2)!!.
$$
\end{lemma}

\begin{proof}
To simplify the notations, set $S_n=\sum\limits_{k=1}^{n-2} k^2 f_{k,n}$.
As we have seen in the proof of Lemma \ref{lem:f},
$$
f_{k,n}=\frac{(n-2)!k}{2^{n-k-2}}\sum_{m=0}^{n-k-2} \frac{4^m (2n-2m-k-5)!}{(n-m-2)!(n-m-k-2)!}
$$
and therefore
$$
\begin{array}{rl}
S_n &\displaystyle= \frac{(n-2)!}{2^{n-2}} \sum_{k=1}^{n-2} 2^k k^3 \sum_{m=0}^{n-k-2} \frac{4^m (2n-k-2m-5)!}{(n-k-2)! (n-k-m-2)!} \\ & \displaystyle =
\frac{(n-2)!}{2^{n-2}} \sum_{k=1}^{n-2} 2^k k^3 \sum_{m=0}^{n-k-2} \frac{4^{n-k-2-m} (k+2m-1)!}{(k+m)! m!} \\ & \displaystyle =
(n-2)! 2^{n-2} \sum_{k=1}^{n-2} \frac{k^3}{2^k} \left(\frac{1}{k}+\sum_{m=1}^{n-k-2} \frac{1}{4^m m}\binom{k+2m-1}{k+m} \right) \\ & \displaystyle =
(n-2)! 2^{n-2} \left(6-\frac{n^2+2}{2^{n-2}}+ \sum_{k=1}^{n-2}\frac{k^3}{2^k}\sum_{m=1}^{n-k-2} \frac{1}{4^m m}\binom{k+2m-1}{k+m} \right)
\end{array}
$$
Set now 
$$
\displaystyle S_n'= \sum_{k=1}^{n-2}\frac{k^3}{2^k}\sum_{m=1}^{n-k-2} \frac{1}{4^m m}\binom{k+2m-1}{k+m}=\sum_{k=1}^{n-3}\frac{k^3}{2^k}\sum_{m=1}^{n-k-2} \frac{1}{4^m m}\binom{k+2m-1}{k+m}
$$
Since $S_3'=0$, we have that
$$
S_n'  =\sum_{p=3}^{n-1} (S_{p+1}'-S_p')
$$
and
$$
\begin{array}{l}
\displaystyle S_{p+1}'-S_p'=\frac{(p-2)^3}{2^p}+\sum_{k=1}^{p-3} \frac{k^3}{2^k(p-k-1)4^{p-k-1}}\binom{2p-k-3}{p-1}\\
\qquad \displaystyle
=\frac{(p-2)^3}{2^p}+\frac{1}{2^{2p-2}} \sum_{k=1}^{p-3} \frac{k^3(2p-k-3)!}{2^{-k}(p-k-1) (p-1)!(p-k-2)!}\\
\qquad \displaystyle
=\frac{(p-2)^3}{2^p}+\frac{1}{2^{2p-2}(p-1)!} \sum_{k=1}^{p-3} \frac{k^3(2p-k-3)!}{2^{-k}(p-k-1)!}\\
\qquad \displaystyle
=\frac{(p-2)^3}{2^p}+\frac{1}{2^{2p-2}(p-1)!} \sum_{k=1}^{p-3} \frac{(p-k-2)^3(p+k-1)!}{2^{k-p+2}(k+1)!}\\
\end{array}
$$
$$
\begin{array}{l}
\qquad \displaystyle
=\frac{(p-2)^3}{2^p}+\frac{1}{2^{p-1}(p-1)!} \sum_{k=2}^{p-2} \frac{(p-k-1)^3(p+k-2)!}{2^{k}k!}\\
\qquad \displaystyle
=\frac{(p-2)^3}{2^p}+\frac{1}{2^{p-1}(p-1)!} \Big[\sum_{k=0}^{p-2} \frac{(p-k-1)^3(p+k-2)!}{2^{k}k!}\\
\qquad\qquad\qquad \displaystyle-(p-1)^3(p-2)!-\frac{1}{2}(p-2)^3(p-1)!\Big]
\\
\qquad \displaystyle
=-\frac{(p-1)^2}{2^{p-1}}+\frac{1}{2^{p-1}(p-1)!} \sum_{k=0}^{p-2} \frac{(p-k-1)^3(p+k-2)!}{2^{k}k!}
\\
\qquad \displaystyle
=-\frac{(p-1)^2}{2^{p-1}}+\frac{1}{(2p-2)!!}\big((4p-1)(2p-3)!!-3(2p-2)!!\big)\quad\mbox{(by Lemma \ref{sum:d})}
\\
\qquad \displaystyle
=-\frac{(p-1)^2}{2^{p-1}}+(4p-1)\frac{(2p-3)!!}{(2p-2)!!}-3
 \end{array}
$$

Therefore
$$
S_n'  =\sum_{p=3}^{n-1}\Big((4p-1)\frac{(2p-3)!!}{(2p-2)!!}-\frac{(p-1)^2}{2^{p-1}}-3\Big)
$$
Now, applying \emph{Gosper's algorithm}   \cite[p. 77]{AeqB} we have that
$$
\sum_{p=3}^{n-1}  (4p-1)\frac{(2p-3)!!}{(2p-2)!!}=\frac{1}{3\cdot 2^{2n+1}}\Big(32(4n^2-3n-1)\binom{2n-3}{n-1}-39\cdot 2^{2n}\Big)
$$
and then
$$
\begin{array}{rl}
S_n'   & \displaystyle = 
\frac{1}{3\cdot 2^{2n+1}}\Big(32(4n^2-3n-1)\binom{2n-3}{n-1}-39\cdot 2^{2n}\Big)\\
& \displaystyle \qquad\qquad -\frac{11\cdot 2^n-8(n^2+2)}{2^{n+1}}
-3(n-3)\\
& \displaystyle = \frac{n^2+2}{2^{n-2}} -3 (n+1) + \frac{(4n+1) (2n-3)!!}{3 (2n-4)!!}.
\end{array}
$$
Finally,  
$$
\begin{array}{rl}
S_n & \displaystyle = (n-2)! 2^{n-2} \left(6-\frac{n^2+2}{2^{n-2}}+S_n'\right)\\
& \displaystyle= -3 (n-1)! 2^{n-2} + \frac{(4n+1)(2n-3)!!}{3}\\
& \displaystyle  =\frac{1}{3} (4n+1) (2n-3)!!-\frac{3}{2} (2n-2)!!.
\end{array}
$$
\end{proof}

Summarizing,  by Lemmas \ref{lem:9}, \ref{sum:d}, and \ref{sum:f}, we have that
$$
\begin{array}{l}
\displaystyle E_U(\oFd_n)=\frac{1}{(2n-3)!!}\Big(n\sum_{k=1}^{n-1} k^2\cdot d_{k,n}+\binom{n}{2}\sum_{k=1}^{n-2} k^2\cdot f_{k,n}\Big)\\
\qquad \displaystyle =\frac{1}{(2n-3)!!}\Big[n((4n-1) (2n-3)!! -3  (2n-2)!!)\\
\qquad \displaystyle\qquad\qquad +\binom{n}{2}\Big(\frac{1}{3} (4n+1) (2n-3)!!-\frac{3}{2} (2n-2)!!\Big)\Big]\\
\qquad \displaystyle = \frac{1}{6} n (4 n^2+21 n-7)-\frac{3n(n+3)}{4}\cdot \frac{(2n-2)!!}{(2n-3)!!}
\end{array}
$$
as we claimed.
%
%


\begin{thebibliography}{11}

\bibitem{AS}
M. Abramowitz, I. Stegun,
\textsl{Handbook of Mathematical Functions with Formulas, Graphs, and Mathematical Tables}. Dover (1964). 

\bibitem{Ald1}
D. Aldous. Probability distributions on cladograms. Random
discrete structures, IMA Vol.
Math. Appl. 76 (Springer,1996), 1--18. 

%

\bibitem{BF06} M. G. B. Blum, O. Fran\c cois, Which random processes describe the Tree of Life? A large-scale study of phylogenetic tree imbalance. Sys. Biol. 55 (2006), 685--691.

%

\bibitem{Brown} J. Brown, Probabilities of evolutionary trees. Syst. Biol. 43 (1994), 78--91.

\bibitem{BS09} D. Bryant, M. Steel, Computing the distribution of a tree metric. IEEE/ACM Trans. Comp. Biol. Bioinf. 16 (2009), 420--426.

\bibitem{Yulenodal}
G. Cardona, A. Mir, F. Rossell\'o, The expected value under the Yule model of the squared path-difference distance.
Appl. Math. Let.  25 (2012), 2031--2036.


\bibitem{CMR}
G. Cardona, A. Mir, F. Rossell\'o, Exact formulas for the variance of several balance indices under the Yule model. To appear in J. Math. Bio. \url{http://dx.doi.org/10.1007/s00285-012-0615-9}



\bibitem{copheneticd1}
G. Cardona, A. Mir, L. Rotger, F. Rossell\'o, D. S\'anchez, Cophenetic metrics for phylogenetic trees, after Sokal and Rohlf. BMC Bioinformatics (2013) 14:3 

%
\bibitem{CS} L. L. Cavalli-Sforza,  A. Edwards, Phylogenetic analysis. Models and estimation procedures. Am. J. Hum. Genet., 19 (1967), 233--257.

\bibitem{fel:04} J.~Felsenstein, \textsl{Inferring Phylogenies}. Sinauer Associates Inc., 2004.

\bibitem{Ford}
D. Ford. Probabilities on cladograms: Introduction to the alpha model.
\texttt{arXiv:math/0511246 [math.PR]} (2005).


\bibitem{HF1}
 \url{http://functions.wolfram.com/HypergeometricFunctions/Hypergeometric3F2/03/07/02}.

\bibitem{HF2}
 \url{http://functions.wolfram.com/HypergeometricFunctions/Hypergeometric2F1/03/03/01}.


%
\bibitem{Harding71} E. Harding, The probabilities of rooted tree-shapes generated by random bifurcation. Adv. Appl. Prob. 3  (1971), 44--77.

\bibitem{Heard92} S. B. Heard,
Patterns in Tree Balance among Cladistic, Phenetic, and Randomly Generated Phylogenetic Trees.
Evolution
46  (1992),  1818--1826.

\bibitem{Hendy} M.  Hendy, C.  Little,   D. Penny, Comparing Trees with
Pendant Vertices Labelled. SIAM J. Applied Math. 44 (1984),  1054--1065.


\bibitem{MirR10} 
A. Mir, F. Rossell\'o, The mean value of the squared path-difference distance for rooted phylogenetic trees.
J. Math. Anal. Appl.  371 (2010),  168--176.
  

\bibitem{MRR} 
A. Mir, F. Rossell\'o, L. Rotger, A new balance index for phylogenetic trees. Math. Biosc. 241 (2013), 125--136.

\bibitem{Mooers97}
A. Mooers, S. B. Heard, Inferring evolutionary process from phylogenetic tree shape. Quart. Rev. Biol. 72  (1997), 31--54.

\bibitem{treebase} V. Morell,  TreeBASE: the roots of phylogeny. Science 273 (1996), 569--560. \url{http://www.treebase.org}

\bibitem{AeqB} 
M. Petkovsek, H. Wilf, D. Zeilberger, $A=B$. AK Peters Ltd. (1996).
Available online at \url{http://www.math.upenn.edu/~wilf/AeqB.html}.


%

%
\bibitem{Rosen78} D. E. Rosen, Vicariant Patterns and Historical Explanation in Biogeography.
Syst. Biol. 27 (1978), 159--188.

%
\bibitem{Sackin:72} M. J. Sackin, ``Good'' and ``bad'' phenograms. Sys. Zool, 21 (1972), 225--226.


\bibitem{Sokal:62} R. Sokal, F. Rohlf, 
The Comparison of Dendrograms by Objective Methods.
Taxon 11 (1962),  33--40.

\bibitem{Steel88} M. Steel,  Distribution of the symmetric difference metric on phylogenetic trees. SIAM J. Discr. Math. 1 (1988),  541--551.

%
\bibitem{cherries} M. Steel, A. McKenzie, Distributions of cherries for two models of trees. Math. Biosc. 164 (2000), 81--92.

%
\bibitem{SM01} M. Steel, A. McKenzie, Properties of phylogenetic trees generated by Yule-type speciation models. Math. Biosc. 170 (2001), 91--112.

%
\bibitem{steelpenny:sb93}
M.~A. Steel, D.~Penny, Distributions of tree comparison metrics---some new
  results, Syst.\ Biol. 42~(2) (1993) 126--141.


\bibitem{Yule} G. U. Yule, A mathematical theory of evolution based on the conclusions of Dr J. C. Willis. Phil. Trans.   Royal Soc. (London) Series B 213 (1924), 21--87.

\end{thebibliography}
\end{document}